\documentclass{amsart}

\usepackage{amsmath,amsthm,amssymb}
\usepackage{comment} 
\usepackage{xr}
\usepackage{multicol}
\usepackage{mdframed}
\usepackage{blkarray}
\usepackage{mathrsfs}
\usepackage{placeins}
\usepackage[table,usenames,dvipsnames]{xcolor}  
\usepackage{tikz}         
\usetikzlibrary{calc}    
\usepackage[most]{tcolorbox}

\usepackage[english]{babel}
\DeclareMathAlphabet{\mathpzc}{OT1}{pzc}{m}{it}

\DeclareMathOperator{\ass}{ass}
\makeatletter
\def\moverlay{\mathpalette\mov@rlay}
\def\mov@rlay#1#2{\leavevmode\vtop{%
    \baselineskip\z@skip \lineskiplimit-\maxdimen
    \ialign{\hfil$\m@th#1##$\hfil\cr#2\crcr}}}
\newcommand{\charfusion}[3][\mathord]{
  #1{\ifx#1\mathop\vphantom{#2}\fi
    \mathpalette\mov@rlay{#2\cr#3}
  }
  \ifx#1\mathop\expandafter\displaylimits\fi}
\DeclareRobustCommand\bigop[1]{%
  \mathop{\vphantom{\sum}\mathpalette\bigop@{#1}}\slimits@
}
\newcommand{\bigop@}[2]{%
  \vcenter{%
    \sbox\z@{$#1\sum$}%
    \hbox{\resizebox{\ifx#1\displaystyle.9\fi\dimexpr\ht\z@+\dp\z@}{!}{$\m@th#2$}}%
  }%
}
\makeatother

\newtcolorbox{BoxDM}[2][]{
  lower separated=false,
  colback=white,
  boxrule=0.4pt,     
  colframe=black,fonttitle=\bfseries,
  colbacktitle=black,
  coltitle=white,
  enhanced,
  attach boxed title to top left={yshift=-0.1in,xshift=0.15in},
  boxed title style={boxrule=0pt,colframe=white,},
  title=#2,#1}

\newcommand{\cupdot}{\charfusion[\mathbin]{\cup}{\cdot}}

\newcommand{\coloneqq}{\mathrel{:=}} 

\newtheorem{lemma}{\bf Lemma}[section]
\newtheorem{proposition}{\bf Proposition}[section]

\newtheorem{definition}{\bf Definition}[section]

\begin{document}

\title{Assembly in Directed Hypergraphs}

\author{
Christoph Flamm$^{1}$, Daniel Merkle$^{2,3}$ and Peter F.\ Stadler$^{4-7}$}


\address{
  $^{1}${Department of Theoretical Chemistry, University of Vienna,
    Austria}\\
  $^{2}${Algorithmic Cheminformatics Group, Faculty of Technology;
    Center for Biotechnology (CeBiTec), Faculty of Biology,
    Bielefeld University, Bielefeld}\\
  $^{3}${Dept. of Mathematics and Computer Science,
    University of Southern Denmark, Odense M, Denmark}\\ 
  $^{4}${Max Planck Institute for Mathematics in the Sciences, Leipzig,
    Germany}\\
  $^{5}${Bioinformatics Groups, Department of Computer Science;
    Competence Center for Scalable Data Services and Solutions
    Dresden/Leipzig (scaDS.AI); German Centre for Integrative Biodiversity
    Research (iDiv); Leipzig Research Center for Civilization Diseases;
    School of Embedded and
    Compositive Artificial Intelligence, Universit{\"a}t Leipzig, Germany}\\
  $^{6}${Facultad de Ciencias, Universidad National de Colombia,
    Sede Bogot{\'a}, Colombia}\\
  $^{7}${Santa Fe Institute, Santa Fe, NM}
}


\keywords{assembly index; retrosynthesis, directed hypergraphs,
  B-hypergraphs, hyperpath problems, graph rewriting}


\begin{abstract}
  Assembly theory has received considerable attention in the recent past.
  Here we analyze the formal framework of this model and show that assembly
  pathways coincide with certain minimal hyperpaths in B-hypergraphs. This
  makes it possible to generalize the notion of assembly to general
  chemical reaction systems and to make explicit the connection to rule
  based models of chemistry, in particular DPO graph rewriting. We observe,
  furthermore, that assembly theory is closely related to retrosynthetic
  analysis in chemistry. The assembly index fits seamlessly into a large
  family of cost measures for directed hyperpath problems that also
  encompasses cost functions used in computational synthesis planning. This
  allows to devise a generic approach to compute complexity measures
  derived from minimal hyperpaths in rule-derived directed hypergraphs
  using integer linear programming.
\end{abstract}

\maketitle


\section{Introduction}

Assembly theory models the combinatorial generation of innovation in
rule-based, abstract universes of recombining objects, where it seeks to
quantify the intrinsic complexity of composite objects in terms of the
minimum number of steps needed for its assembly or construction from basic
building blocks \cite{Marshall:17,Marshall:21,Liu:21,Marshall:22}. As
summarized in \cite{Jaeger:24}, it combines a rule-based system with an
abstracted temporal dimension in the form of recursivity, i.e., by
allowing, in each step of the assembly process, the reuse of all previously
assembled objects.

The theory of constructive dynamical systems \cite{Banzhaf:04,Dittrich:07}
was pioneered already three decades ago in the AlChemy model
\cite{Fontana:94b,Fontana:94c}. In AlChemy, the $\lambda$-calculus is used
to capture the constructive feature of chemistry and the fact that many
different reactants can yield the same stable product. The objects are
$\lambda$-expression that ``react'' with each other by application and
subsequent reduction to normal form. Intended as an abstract model for the
emergence of biological organizations, the model was primarily used to
demonstrate the emergence of system-level features such as
self-maintenance. AlChemy also formed the conceptual basis for a graph-based
artificial chemistry in which chemical reactions are represented by the
application of graph-rewriting rules \cite{Benkoe:03b}.  Chemical reactions
are then very naturally represented as edges in a directed hypergraph
\cite{Zeigarnik:00}. The software package \texttt{M{\O}D}
\cite{Andersen:16a} was developed as an efficient implementation of this
idea, combining rule based DPO graph rewriting \cite{Andersen:13a}, a
framework of strategies from exploration \cite{Andersen:14a}, and
facilities to analyze reaction networks in terms integer hyperflows
\cite{Andersen:19a}.

The notion of assembly in the sense of \cite{Marshall:21,Sharma:23} is
closely related to the problems in chemical synthesis planning (see
\cite{Wang:21} for a recent review), introduced as a computational problem
already half a century ago in the form of retrosynthetic analysis
\cite{Corey:69}, see also \cite{Corey:75,Hendrickson:81}.  Starting from
the target, chemical structures are recursively split into smaller
molecules until sufficiently simple or commercially available starting
materials have been found. The choice of the chemical bonds to be split are
usually chosen to mimic the choices of chemists. Alternatively, descriptors
such as the graph theoretical molecular complexity of the resulting
precursor can be minimized \cite{Bertz:93}. The cost of synthesizing a
given quantity of the target compound then may serve as a measure of the
synthetic difficulty of the target.

A formal exposition of assembly in the sense of
\cite{Marshall:17,Marshall:21,Liu:21,Sharma:23} is given in terms of
certain edge-labeled multigraphs \cite{Marshall:22}.  In
section~\ref{ssect:ass2Bhyp} we shall see that it is not difficult to
translate the construction to the framework of directed hypergraphs. The
notion of assembly pathyways coincides with shortest B-hyperpaths (in the
sense of \cite{Ausiello:05}). This makes it possible to embed the
constructions of assembly theory, such as assembly pathways, as a special
case into the framework of constructive models as representations of
chemistry. Interestingly, the same hypergraph constructions have also been
used in the context of assembly processes in fabrication \cite{Gallo:02}.

Assembly theory also has close connections to the theory of algorithmic
information theory \cite{Chaitin:66,Zvonkin:70}. The assembly index, in
fact, fits seamlessly with a diversity of alternative measures of
complexity listed in \cite{Lloyd:01}. Of particular relevance for our
purposes is grammar-based compression \cite{Kieffer:00}, which for a given
string $\xi$ seeks to construct a \emph{straight-line program} for $\xi$,
i.e., a deterministic context free grammar (CFG) $\mathfrak{G}_x$ that
generates a language consisting only of $\xi$ with as few derivation rules
as possible. Thm.~7 in \cite{Abrahao:24} shows that every minimal assembly
pathways can be expressed (using suitable string encodings of the objects)
by a CFG with derivation rules that correspond to the individual assembly
steps. Moreover, it expresses the assembly intermediates by corresponding
subsets of rules. Assembly, therefore, can also be regarded as a grammar
compressor with derivation rules defined by the underlying physical,
chemical, or biological processes. The Lempel-Ziv (LZ) family of
compression algorithms \cite{Ziv:78} are well-known grammar-based
compressors \cite{Rytter:04}. It is argued in \cite{Abrahao:24} that the
assembly index, which is difficult to compute in general, could (and
should) be replaced by the LZ compression length, since the latter is much
easier to compute. This approach, however, sacrifices the link to the
physical or chemical basis of the measure, such as sequences of chemical
reactions, which -- depending on the application -- may not be desirable.
Compression-based measures, however, have been shown in
\cite{Uthamacumaran:22}.to outperform assembly measures \cite{Marshall:17}
for the task of discriminating molecules of biotic and abiotic origin, see
also \cite{Hazen:24}.

Nevertheless, the concept of assembly is of key interest in many
applications, not despite but \emph{because} it is dependent on a specific
model of objects and their construction. In synthetic chemistry, for
example, it is of utmost practical importance how difficult it is to
synthesize a putative drug target making use only of commercially viable
reactions and starting materials. In this setting, an estimate of the
algorithmic complexity of the target, computed using some efficient string
compression algorithm from a string encoding of the molecular structure such
as canonical SMILE, is of very little use. Whether or not compression
length, the molecular assembly index (MA) \cite{Marshall:17,Jirasek:23}, or
another particular measure of complexity is useful, and which measure is
superior, necessarily remains dependent on the pertinent (research)
question.

Instead of advocating for a particular measure, we are concerned here with two
issues: (1) The formalism described in \cite{Marshall:22} does not fit well
with the mathematical framework of directed hypergraphs that is most
commonly utilized for studying chemical reaction networks. This also
concerns further extensions e.g.\ to Petri nets \cite{Andersen:23a}. In
Sect.~\ref{sect:assH}, we therefore start from a translation of assembly
theory to acyclic directed B-hypergraphs and show that the requirement of
acyclicity for the underlying space can be dropped by identifying assembly
pathways with minimal B-hyperpaths. The connection between assembly and
grammar-based compression emphasized in \cite{Abrahao:24} also persists in
this much more general setting. Finally, we show that assembly can be
studied in the general setting of directed hypergraphs and fits seamlessly
into the class of path problems on directed hypergraphs.

In Section~\ref{sect:comp} we consider computational approaches and show
that the rule-based aspect of assembly is captured in a very natural way by
graph transformation systems, and in particular by DPO graph rewriting. The
invertibility of DPO rules guarantees that assembly can be computationally
treated as the disassembly of the target object. The resulting procedure is
akin retrosynthetic analysis, which was introduced a theoretical foundation
of chemical synthesis planning already in the 1960s
\cite{Vleduts:63,Corey:67}. We then show that, given a directed hypergraph
produced by a generic transformation systems, it is possible to construct
optimal pathways by means of integer linear programming (ILP). The cost
functions amenable to this approach include the assembly index as well as
yield-based synthesis costs.

\section{Assembly in Hypergraphs} 
\label{sect:assH}

\subsection{From Assembly Spaces to Acyclic B-Hypergraphs}
\label{ssect:ass2Bhyp}

The notion of an \emph{assembly space} is introduced in \cite{Marshall:22}
starting from certain edge-labeled multigraphs. A similar construction is
used in \cite{Bieniawski:24,Lukaszyk:25}. While directed multigraphs, also
known as quivers \cite{Derksen:05}, are certainly appealing from the point
of view of category theory and algebra, they need to be endowed with an
additional edge labeling function in the context of assembly theory that
can be inconvenient two work with. In this section we briefly review the
original construction of assembly spaces as presented in \cite{Marshall:22}
and show that it can be replaced by a conceptually simpler representation
in terms of a restricted class of directed (multi-)hypergraphs. This not
only slightly simplifies some technicalities but also readily lends itself
to natural generalizations and connects assembly theory to a body of
earlier work on hypergraph algorithms.

\paragraph{Quivers, reachability, and assembly spaces.}
A directed multigraph (or quiver) $\Gamma=(V,Q,h^-,h^+)$ consists of a (not
necessarily finite) vertex set $V$, an edge set $Q$, and maps $h^-: Q\to
V$ and $h^+: Q\to V$ identifying the tail $h^-(e)$ and head $h^+(e)$ of
every edge in $Q$. An edge $e\in Q$ is a loop if $h^+(e)=h^-(e)$. A path in
$\Gamma$ is a sequence of $(x_0,e_1,x_1,e_2,\dots,e_n,x_n)$ such that
$h^-(e_i)=x_{i-1}$ and $h^+(e_i)=x_{i}$ for $1\le i\le n$. The length of
the path is $n$, and we also consider $(x_0)$ to be a path (with length $0$).

For $x,y\in V$ we say $y$ is reachable from $x$, in symbols $x\le y$, if
there is a path from $x$ to $y$. Note that we have $x\le x$ for all $x\in
V$. Since paths can be concatenated, the reachability relation is
transitive, i.e., $x\le y$ and $y\le z$ implies $x\le z$.  A cycle in
$\Gamma$ is a path of length $n\ge 1$ such that $x_0=x_n$. Note that the
path $(x_0)$ is not a cycle, while a loop $(x_0,e_1,x_0)$ counts as a
cycle. A directed multigraph is \emph{acyclic} if there are no
cycles. Acyclic directed multigraphs in particular do not contain loops.
The relation $\le$ is antisymmetric, $x\le y$ and $y\le x$ implies $x=y$,
if and only if there are no cycles with length larger than $1$. In particular,
therefore, $\le$ is a partial order on $V$ for acyclic directed
multigraphs.

A vertex $x\in V$ is minimal if there is no $z\in V\setminus\{x\}$ such
$z\le x$. Equivalently, $x$ is minimal if, for all $z\in V$, $h^+(z)=x$
implies $z=x$. We write $\min V \coloneqq \{y\in V| z\le y \implies z=y\}$
and $\max V \coloneqq \{y\in V| y\le z \implies z=y\}$ for the sets of
minimal and maximal vertices in $\Gamma$, respectively. We say that
$\Gamma$ is \emph{grounded} if for every $x\in V$ there is a $z\in\min V$
such that $z\le x$.

Using the language introduced in the paragraphs above, an \emph{assembly
space} is defined in \cite{Marshall:22}, as a grounded, acyclic directed
multigraph with finite basis set $\min V$, endowed with an additional edge
labeling function $\varphi:Q\to V$ such that for every edge $a$ from $x$ to
$z$ with $\varphi(a)=y$ there is an edge $b$ from $y$ to $z$ with
$\varphi(b)=x$.

\paragraph{Directed hypergraphs.}
In the following we use the notion of directed hypergraphs introduced in
\cite{Gallo:93}. We will, however, consider edges as pairs of multisets
instead of sets to accommodate multiplicities. For multisets we abuse set
notation and write $x\in A$ to mean that $x$ appears in $A$ with
multiplicity at least $1$. We also write $|A|$ for the sum of the
multiplicities of the $x\in A$ and $A\subseteq B$ to mean that every
element of $A$ is contained in $B$ with at least the same multiplicity.

A directed (multi)hypergraph $\mathscr{H}$ consists of a set $V$ of
vertices and (multi)set $\mathcal{E}$ of directed hyperedges of the form
$(E^-,E^+)$ where $E^-$ and $E^+$ are (multi)sets comprising elements of
$V$. A hyperpath in $\mathscr{H}$ is a sequence
$(x_0,E_1,x_1,\dots,E_n,x_n)$ such that $x_{i-1}\in E_{i}^-$ and
$x_{i}\in E_{i}^+$ for $1\le i\le n$.  A hyperpath is cyclic if
$x_n\in E_1^-$. A directed hypergraph without cyclic hyperpath is acyclic.
A subhypergraph $\mathscr{H}'=(V',\mathcal{E}')$ of
$\mathscr{H}=(V,\mathcal{E})$ is a hypergraph such that $V'\subseteq V$ and
$\mathcal{E'}\subseteq\mathcal{E}$. 

\begin{figure}
  \begin{tabular}{ccc}
  \begin{minipage}{0.5\textwidth}
    \includegraphics[width=\textwidth]{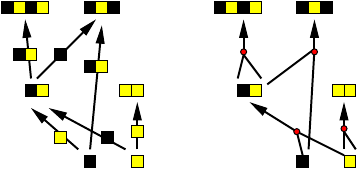} 
  \end{minipage} &&
  \begin{minipage}{0.5\textwidth}
    \caption{Correspondence of assembly spaces as defined in
      \cite{Marshall:22} as labeled multigraphs and as B-hypergraphs with
      tails of cardinality $2$. The latter are drawn in their bipartite
      K{\"o}nig representation, with red circles denoting the hyperedges.}
  \end{minipage}
  \end{tabular}
\end{figure}

Denote by $E_{xy:z}$ the set of all edges $e$ in the assembly space from
$x$ to $z$ with label $\varphi(e)=y$ and all edges $f$ from $y$ to $z$ with
$\varphi(e)=x$. By definition of $\varphi$, $E_{xy:z}$ contain at least one
edge from $x$ to $z$ and one edge from $y$ to $z$. Note that
$E_{xy:z}=E_{yx:z}$ and $E_{xx:z}$ contains exactly the edges $e$ from $x$
to $z$ with $\varphi(e)=x$. Since acyclicity implies that there are no
loops, we have $E_{xy:x}=\emptyset$. Finally,
$E_{xy:z}\cap E_{x'y':z'}=\emptyset$ unless $\{x,y\}=\{x',y\}$ and $z=z'$.

Consider the hypergraph $\mathscr{H}=(V,\mathcal{E})$ with vertex set $V$
and hyperedges $(\{x,y\},\{z\})\in\mathcal{E}$ iff and only if
$E_{xy:z}\ne\emptyset$. To account for the fact $E_{xx:z}$ is meant to
account for assembly of $z$ by merging two copies of $x$, we represent the
corresponding hyperedge by $(\{x,x\},\{z\})$ instead of $(\{x\},\{z\})$.
By construction, $E_{xy:z}\ne\emptyset$ if and only if there is an edge
from $x$ to $z$ in the assembly space.

A directed (multi)hypergraph $\mathscr{H}$ is a B-hypergraph if its
directed edges $(E^-,E^+)$ satisfy $|E^+|=1$, i.e., the heads of all
hyperedges are single vertices \cite{Gallo:93}. The edge-labeled
multi-hypergraphs of \cite{Marshall:22} thus give rise to B-hypergraphs
with $|E^-|=2$ for all edges. Conversely, we can take any B-hypergraphs on
$V$ with $|E^-|=2$ and construct an edge-labeled directed multigraph by
inserting for $(\{x,y\},\{z\})$ the edges $e$ between $x$ and $z$ with
label $\varphi(e)=y$ and $f$ between $y$ and $z$ with label $\varphi(e)=x$
provided $x\ne y$, and a single edge $e$ between $x$ and $z$ with label
$\varphi(e)=x$. Clearly this satisfied the labeling condition required for
assembly spaces in \cite{Marshall:22}.

Thus, for all $x,y\in V$ we have $x\le y$ (in the assembly space) if either
$x=y$ or there is a hyperpath from $x$ to $z$ in the B-hypergraph
$\mathscr{H}$. Thus reachability in $\mathscr{H}$ is the same as
reachability in the edge-labeled multigraph. Moreover, the edge labeled
multigraph is acyclic if and only if the corresponding directed hypergraph does
not contain a cyclic path and thus is acyclic. The conditions for being
grounded and having a finite set $\min V$ are identical in both
construction by virtue of having the same reachability relation $\le$ on
$V$. Finally it follows immediately from the definitions that assembly
subspaces correspond to subhypergraphs and \emph{vice versa}. The
discussion in this section thus can be summarized as follows: 
\begin{proposition}
  There is a 1-1 correspondence preserving groundedness and reachability
  between assembly spaces in the sense of \cite{Marshall:22} and acyclic
  subhypergraphs of B-hypergraphs satisfying $|E^-|=2$ for all hyperedges.
\end{proposition}
This result not only translates the notion of assembly spaces to the more
convenient framework of directed hypergraphs but also immediately suggests
natural generalizations by relaxing the conditions on the directed
hypergraphs.

\subsection{Assembly on B-Hypergraphs}

In this section we drop the restriction to B-hypergraphs. The definition of
reachability remains unchanged. For later reference:
\begin{definition}
  Let $\mathscr{H}=(V,\mathcal{E})$ be a directed (multi)hypergraph.  Then
  $x$ is reachable from $y$, in symbols $y\le x$ if either $x=y$ or $x\ne
  y$ and there is a path $(y,E_1,x_1,\dots,E_n,x)$ of length $n\ge1$.
\end{definition}
Reachability, which forms the basis of the formal developments below, is
independent of the multiplicity of vertices in the hyperedges. To simplify
the language below we will therefore simply speak of directed hypergraphs.

The fact that $\le$ is a partial order if and only if $\mathscr{H}$ is
acyclic remains true for directed hypergraphs in general.  If $\mathscr{H}$
is acyclic then every subhypergraph $\mathscr{H}'$ of $\mathscr{H}$ is also
acyclic because the reachability relation $\le'$ on $\mathscr{H}'$ is a
subset of $\le$.  A subhypergraph $\mathscr{H}'$ of $\mathscr{H}$ in
general does not inherit the reachability relation from $\mathscr{H}$:
While for all $x,y\in V'$ we have $x\le' y$ implies $x\le y$ in
$\mathscr{H}$, the converse is not true. Note that the property of being
grounded depends on the reachability relation $\le'$ within $\mathscr{H}'$,
not on the reachability relation $\le$ of the parent hypergraph
$\mathscr{H}$.  Every \emph{finite} subhypergraph $\mathscr{H}'$ of
$\mathscr{H}$ is necessarily grounded.

In analogy with the definition in \cite{Marshall:22} we say that a
subhypergraph $\mathscr{H}'$ is \emph{rooted in $\mathscr{H}$} if
$\mathscr{H}'$ is a grounded subhypergraph of $\mathscr{H}$ and $\min
V'\subseteq \min V$.

\begin{definition}
  Let $\mathcal{H}=(V,\mathcal{E})$ be a directed hypergraph with
  reachability relation $\le$. For every $x\in V$ define
  \begin{equation*}
    \mathcal{E}_{\le x}\coloneqq
    \left\{ (E^-,E^+) \,\big|\, E\in\mathcal{E},\,
       \forall z\in E^-: z\le x,\, \text{and}\,
       \exists y\in E^+: y\le x \right\}
  \end{equation*}
  Moreover, we set
  $\displaystyle V_{\le x}\coloneqq \bigcup_{E\in\mathcal{E}_{\le x}}
  (E^-\cup E^+)$ and $\displaystyle V^-_{\le x}\coloneqq
  \{x\}\cup\bigcup_{E\in\mathcal{E}_{\le x}} E^-$.
\end{definition}
The set $\mathcal{E}_{\le x}$ thus contains all edges along paths that
eventually reach $x$. Since $V_{\le x}$ comprises all vertices in the edge
set $\mathcal{E}_{\le x}$, $\mathscr{H}_{\le}[x]\coloneqq (\mathcal{E}_{\le
  x},V_{\le x})$ is a subhypergraph rooted in $x$. In a B-hypergraph, the head
$E^+$ of every edges $E\in\mathcal{E}$ is a single vertex, and thus for
B-hypergraphs we have $V_{\le x}=V'_{\le x}$ for all $x\in V$.

The following result can bee seen as a proper generalization of
\cite[Lemma~5]{Marshall:22}:
\begin{lemma}
  Let $\mathscr{H}=(V,\mathcal{E})$ be a grounded acyclic directed
  hypergraph. Then $\mathscr{H}_{\le}[x]$ is a subhypergraph rooted in
  $\mathscr{H}$. Moreover, the reachability relations of $\mathscr{H}$ and
  $\mathscr{H}_{\le}[x]$ coincide on $V^-_{\le x})$.
\end{lemma}
\begin{proof}
  By construction $z\in E^-,E^+$ implies $z\in V_{\le}$ for all
  $E\in\mathcal{E}_{\le}$ and thus $\mathscr{H}_{\le}[x]$ is a
  subhypergraph of $\mathscr{H}$.  Now consider $y\in V_{\le}$. Since
  $\mathscr{H}$ be a grounded, there is $z\in\min V$ with $z\le y$ and by
  construction $z\in V_{\le}$. By definition of $\min V$, there is no $u\in
  V\setminus\{z\}$ with $u\le z$, thus $z\in \min V_{\le}$.  Thus
  $\mathscr{H}_{\le}[x]$ is grounded. If $y\notin \min V$, then $z\ne y$,
  and thus $y\notin \min V_{\le}$. Therefore $\min V_{\le}\subseteq \min
  V$, and $\mathscr{H}_{\le}[x]$ is rooted in $\mathscr{H}$.

  Denote by $\le'$ the reachability relation on
  $\mathscr{H}_{\le}[x]$. Since $\mathscr{H}_{\le}[x]$ is a subhypergraph
  of $\mathscr{H}$, $u\le' v$ implies $u\le v$. Conversely, let $u,v\in
  V'_{\le x}$ and consider a path $\mathbf{p}$ from $u$ to $v$ in
  $\mathscr{H}$ and let $E=(E^-,E^+)$ be an edge along $\mathbf{p}$. Since
  we have either $v=x$ of $v\in E^-$ for some edge $E\in\mathcal{E}$ we
  have $v\le x$ and thus the path $\mathbf{p}$ can be extended to a path
  from $u\to x$; its edges are by construction contained in
  $\mathcal{E}_{\le}$ and thus in particular $\mathbf{p}$ is a path in
  $\mathscr{H}_{\le}[x]$. Therefore $u\le' v$ and the reachability
  relations $\le$ and $\le'$ thus coincide on $V'_{\le}$.
\end{proof}

We are now in the position to define assembly pathways in general acyclic
grounded hypergraphs. By virtue of the correspondence between assembly
spaces \textit{sensu} \cite{Marshall:22} and B-hypergraphs with $|E^-|=2$
for all hyperedges, these definitions are proper generalizations of the
constructions for assembly spaces.

\begin{definition}
  Let $\mathscr{H}=(V,\mathcal{E})$ be a grounded acyclic hypergraph and
  $x\in V$. Then an \emph{assembly pathway} for $x$ is a subhypergraph
  $\mathscr{P}_x=(W,\mathcal{F})$ with reachability relation $\le'$ on $V$
  such that (i) $\mathscr{P}_x$ is rooted in $\mathscr{H}$ and (ii) $z\le'
  x$ holds for all $z\in E^-$ with $(E^-,E^+)\in \mathcal{F}$.
\end{definition}
We immediately see that $\mathscr{P}_x$ is also a subhypergraph rooted in
$\mathscr{H}_{\le}[x]$. The second condition in this definition excludes
reactions from the pathway that have no product connected to the
target $x$ by a path $\mathscr{P}_x$.  

A closely related concept was defined in \cite{Ausiello:05} for
B-hypergraphs as follows: A B-hyperpath $\mathscr{Q}$ from $s$ to $t$ is a
subhypergraph of a B-hypergraph $\mathscr{H}$ such that (i) $\mathscr{Q}$
is minimal (w.r.t.\ deletion of vertices and hyperedges), (ii) there is a
linear order $E_1,E_2,\dots, E_n$ of the hyperedges of $\mathscr{Q}$ such
that
\begin{equation*}
  E_k^- \subseteq \{s\}\cup \bigcup_{j=1}^{k-1} E_j^+ \quad
  \text{for all} 
\end{equation*}
and $t\in E_n$. In B-hypergraphs, the minimality condition (i) implies that
$\mathscr{P}$ is acyclic: if there is a vertex $z$ such that $E_k^+=\{z\}$
and $z\in E^-_j$ for some $k\le j$, then the hyperedge $E_j$ can be deleted
without disturbing property (ii) in the definition. Note that the
reachability partial order $\le'$ on $\mathscr{Q}$ is consistent with the
linear of the edges. Given a subset $S\subseteq V$ of designated ``source
vertices'', an augmented hypergraph $\mathscr{H}'$ can be constructed by
adding an artifical source vertex $o$ and directed edge $(\{o\},\{s\})$ for
all $s\in S$. Then every B-hyperpath $\mathscr{Q}'$ from $o$ to $t\in V$
consists of $o$, a subset of the edges from $o$ to elements of $S$, and an
acyclic subhypergraph $\mathscr{P}$ of $\mathscr{H}$ that is grounded in a
subset of $S$ and has $t$ as its unique maximum element
\cite{Fagerberg:18a}. That is, $\mathscr{P}$ is an assembly pathway. Note
that in this setting, we do not require that $\mathscr{H}$ itself is
acyclic. Conversely, consider an acyclic subhypergraph $\mathscr{P}$ of
$\mathscr{H}$ grounded in the set $S\subseteq V$ of its minimal
elements. Then the augmented subhypergraph $\mathscr{H}'$ constructed as
above contains an B-hyperpath $\mathscr{Q}$ from $o$ to $t$ that is a
subhypergraph of the augmentation $\mathscr{P}'$.

In a B-hypergraph, therefore, minimum size assembly pathways are
essentially the same as minimum size B-hyperpaths. This allows us to drop
the condition that $\mathscr{H}$ is acyclic. Instead of considering
grounded subhypergraphs, it suffices to specify a start subset $S\subseteq
V$. Finding a shortest B-hyperpath in a B-hypergraph is NP-hard in general
\cite{Italiano:89}. It seems unknown, however, whether NP-hardness persists
if tails sizes are bounded by a constant $K$, and particular if $K=2$.

It is also possible to drop the condition that $\mathscr{H}$ is a
B-hypergraph. A natural generalization of B-hyperpaths for this setting is
discussed in \cite{Ritz:17}. In this case, minimality w.r.t.\ edge and
vertex inclusion no longer implies acyclicity. The notion of a
\emph{shortest acyclic B-hyperpath}, however, is still well-defined and
provides a natural way of defining assembly pathways. In \cite{Ritz:17},
this construction is applied to signalling pathways. These authors also
show that computing a shortest B-hyperpath in this setting is NP-hard even
if the size of the heads and tails of the hyperedges is at most $3$.

\subsection{B-hypergraphs and grammar compression}

As noted in \cite{Abrahao:24}, an assembly pathway in a B-hypergraph can be
associated with a CFG that represents each step of the assembly as a
derivation rule.  This observation also pertains to any B-hyperpath
$\mathscr{P'}$ from $o$ to $x$ in an augmented B-hypergraph
$\mathscr{H}'$. As above, we denote by $S\coloneqq \{s|
(\{o\},\{s\})\in\mathcal{E}(\mathscr{P}')\}$ the set of heads of the
augmentation edges, i.e., the basis set of the assembly pathway; we write
$\mathscr{P}$ for the subhypergraph of $\mathscr{P}'$ contained in
$\mathscr{H}$, and set $P\coloneqq V(\mathscr{P})$.  Since $\mathscr{P}'$
is a minimal B-hyperpath, a vertex is the head of exactly one edge, which
we denote by $E(z)=(E^-(z),\{x\})$, if $z\in V\setminus S$, and it does not
appear as the head of an edge at all if $z\in S$. Moreover, recall that the
reachability relation $\le$ on $V$ is a partial order, and thus by
Szpilrajn's theorem \cite{Szpilrajn:30} can be extended to a total order,
which we also denote by $\dot<$. In particular, we can choose this linear
extension such that $z'\dot<z$ if $z'\in S$ and $z\in V\setminus S$.

The following discussion in essence recapitulates Thm.~7 of
\cite{Abrahao:24} with a somewhat simpler specification of the CFG
\cite{Charikar:05} $\mathfrak{G}_x = (S, V\setminus S, x,\mathcal{R})$,
where $S$ is the set of terminal symbols, $V\setminus S$ is the set of
non-terminal symbols, and $x\in V\setminus S$ serves as start symbol.  The
set $\mathcal{R}$ comprises the derivation rules $z\to \tilde E^-(z)$,
here $\tilde E^-(z)=u_1u_2\dots u_q$, $q\ge 1$ denotes the word formed by
the terminals and non-terminals corresponding to the $u_i\in E^-(z)$.  For
every non-terminal $z\in V\setminus S$ there is exactly one derivation
rule. Moreover, transferring the order $\dot<$ from $V$ to the set of
terminals and non-terminals, we have $u_i\dot< z$ for all $u_i\in\tilde
E^-(z)$, i.e., $\mathfrak{G}_x$ is acyclic and hence it generates exactly
one finite string, which represents the composition of $x$ in terms of
elements of $S$. Moreover, restricting the derivations to
$\mathcal{R}_y\coloneqq\{(y\to \tilde E^-(y))| y\le z$ is a grammar that
generates this composition representation of $y$, which by construction
appears as a substring of $x$. The grammar $\mathfrak{G}_x$ thus not only
encodes the target $x$ but the entire sub-hypergraph $\mathscr{P}$.  Since
there is one derivation rule for every non-terminal, we have
$|\mathcal{R}|=|V\setminus S|=\ass(x)$. 

If we insist that $\mathfrak{G}_x$ is in Chomsky normal form, we need $|S|$
additional non-terminals representing the seeds, and $|S|$ additional rules
to replace each of these non-terminals by the corresponding terminal. If
$|E^-(z)|\ge2$ we need to insert $|E^-(z)|-2$ intermediary non-terminals
and corresponding derivation rules. We note, finally, that the contruction
is easily extended to general hypergraphs by inserting derivation rules
$z_i\to \tilde E^-$ for all $z_i\in E^+$ and then removing all but the rule
with $z_i$ on the l.h.s., as well as any rules where $z_i$ is not contained
in the tail of some hyperedge. This amounts to considering $\mathscr{H}^B$
instead of $\mathscr{H}$.

The CFG $\mathfrak{G}_x$ derived directly from the shortest hyperpath will
in general not be the smallest grammar. Finding a smallest grammar,
however, is NP-complete \cite{Charikar:05}. Moreover, the smallest grammar
would correspond to assembly in a setting where essentially any combination
of building block is feasible. This seems unrealistic at least in
applications to chemical systems, and thus may not be the most useful
measure if only a particular subset of transformations -- such as specific
chemical reactions -- is allowed.

\subsection{Assembly on General Directed Hypergraphs} 

The ``complexity'' of an assembly pathway is naturally quantified by the
number of operations, i.e., the number of hyperedges. We introduce the
assembly index here in terms of hypergraphs and show that it coincides with
the definition used in the literature on assembly spaces. 
\begin{definition}
  \label{def:asspath}
  Let $\mathscr{H}=(V,\mathcal{E})$ be a grounded acyclic hypergraph and
  $x\in V$. Then the assembly index $\ass(x)$ of a vertex $x\in V$ is the
  minimum number of hyperedges in an assembly pathway for $x$ in
  $\mathscr{H}$.
\end{definition}
In \cite{Marshall:22} as well as a series of applications,
e.g.\ \cite{Marshall:21,Liu:21,Sharma:23}, the assembly index is defined in
terms of the cardinality of the vertex set of assembly pathways
$\mathscr{P}$. We next show that the two definitions are in fact equivalent
for B-hypergraphs, and thus in particular for assembly spaces:
\begin{lemma}
  If $\mathscr{H}=(V,\mathcal{E})$ is a grounded acyclic B-hypergraph and
  $x\in V$, then $\ass(x)$ equals the number minimum number of non-minimum
  vertices in an assembly pathway $\mathscr{P}$ for $x$.
\end{lemma}
\begin{proof}
  Suppose $\mathscr{P}$ is an assembly pathway for $x$ in $\mathscr{H}$
  that is minimal w.r.t.\ the removal of edges. If $x$ is not reachable
  from $u$, then $\mathscr{P}$ contains a subhypergraph $\mathscr{P}'$ that
  does not contain $u$, and thus every hyperedge $E$ that has $\{u\}=E^+$
  as its head can be removed, contradicting the minimality of
  $\mathscr{P}$. On the other hand, if $\{u\}=E_1^+=E_2^+$ is the head of
  two distinct hyperedges $E_1$ and $E_2$ of $\mathscr{P}$ then either
  $E_1$ or $E_2$ can be removed, resulting in a subhypergraph
  $\mathscr{P}'$ of $\mathscr{P}$. Clearly $\mathscr{P}'$ is an assembly
  pathway since $u$ is reachable from the basis and $x$ is reachable from
  $u$. Again, this contradicts the minimality of $\mathscr{P}$. Thus every
  vertex in $\mathscr{P}$ is the head of at most one edge.  Moreover, every
  vertex in the tail of an edge is either contained in $\min V$ or is the
  head of a unique hyperedge in $\mathscr{P}$. Thus the number of edges in
  $\mathscr{P}$ coincides with the number of vertices in $\mathscr{P}$ that
  are not contained in its basis set $\min V(\mathscr{P})$.
\end{proof}
In B-hypergraphs, therefore, $\ass(x)$ can be expressed the minimum number
of hyperedges in a suitably defined hyperpaths connecting $x$ with the
basis set.  A difficulty arising in the setting of general directed
hypergraphs, is that acyclicity will in general not be attainable if
$|E^+_k|\ge 2$. The reason is that chemical reactions in general also
produce small by-products that may have appeared previously in
$\mathscr{P}$ and thus form cycles.  In particular, acyclicity rules out
any form of network catalysis \cite{Braakman:13,Smith:08}. In the general
case, therefore, it is necessary to allow generalizations of B-hyperpaths
in which heads of edges are not restricted to a single vertex, i.e., for
which $|E^+_k|>1$ is allowed. This forces us to drop the requirement of
acyclicity in Definition~\ref{def:asspath} since we cannot -- and do not
want to -- rule out secondary products.

To accommodate such situations, we first observe that in the case of
$B$-hypergraphs we can re-interpret the notion of reachability of vertices
as reachability of edges. A hyperedge with $E^+=\{x\}$ is reachable through
a pathway $\mathscr{P}$ that reaches $x$ through hyperedges $E_i$ with
$E_i^+=\{x_i\}$ for all $x_i\in E^-$. Thus we may \emph{define} $E'\le E''$
whenever $x'\in E'^+$, $x''\in E''^+$ and $x'\le x''$. Since $\le$ is a
partial order in the if $\mathscr{P}$ is an acyclic B-sub-hypergraph, we
also obtain a partial order on the set of hyperedges in this case.  This
reminiscent of the notion of realizability of chemical reaction pathway in
\cite{Andersen:23a,Andersen:25a}. A sequence of reactions is realizable if
they can be ordered such that for each reaction, all its educts have been
produced already by an earlier reaction. This simple idea naturally
generalizes to assembly processes in general: An assembly step can be
performed only if all vertices in $E^-$ were reached in an earlier step. We
can therefore replace acyclicity of the hypergraph itself by the existence
of a suitable partial order that ensures this type of reachability. We
formalize this idea in the following way:
\begin{definition}
  \label{def:genass} 
  Let $\mathscr{H}$ be a directed hypergraph, $\emptyset\ne S\in
  V(\mathscr{H})$ and $x\in V(H)$.  Then the subhypergraph $\mathscr{P}$ of
  $\mathscr{H}$ is an assembly pathway for $x$ if (i) $\mathscr{P}$ is
  grounded in $S$, (ii) $x\in \max\mathscr{P}$, and (iii) there is a
  partial order $\preceq$ on $E(\mathscr{P})$ such that for all $F\in
  E(\mathscr{P})$ and all $x\in F^-\setminus S$ there is $E_x\in
  E(\mathscr{P})$ with $E_x\preceq F$ and $E_x\ne F$.
\end{definition}
We note that condition (iii) matches the notion of B-connections
\cite{Gallo:93,Nielsen:01,Thakur:09}, except that it starts from a ground
set $S$ instead of a single vertex $s$. Again we may restrict ourselves to
minimal assembly pathways by considering subsets of hyperedges that still
satisfy the definition. Moreover, Def.~\ref{def:genass} reduces to our
previous definition on $B$-hypergraphs.

Given a general directed hypergraph $\mathscr{H}$ we may construct a
B-hypergraph $\mathscr{H}^B$ by replacing every hyperedge $E=(E^+,E^-)$
into a set of $|E^+|$ hyperedges of the form $(E^-,\{x\})$ with $x\in
E^+$. This construction allows us to translate Definition~\ref{def:genass} as
follows:
\begin{lemma}
  \label{lem:general}
  Let $\mathcal{H}$ be a directed hypergraph, let $\mathscr{P}$ be a
  subhypergraph of $\mathscr{H}$ grounded in $S$, and $x\in
  V(\mathscr{H})$.  Then $\mathscr{P}$ is an assembly pathway for $x$ if
  and only if $\mathscr{P}^B$ contains an assembly pathway $\mathscr{P}^*$
  for $x$ on $\mathscr{H}^B$.
\end{lemma}
\begin{proof}
  Suppose $\mathscr{P}$ is an assembly pathway according to
  Definition~\ref{def:genass}, and let $\prec$ be the partial order on
  $E(\mathscr{P})$. Then $\mathscr{P}^B$ inherits a partial order on is
  $B$-hyperedges such that $E\prec F$ translates to
  $(E^-,\{z\})\prec_B(F^-,\{y\})$ for all $z\in E^+$ and $y\in
  E^+$. Consider the subhypergraph $\mathscr{P}^*$ obtained from
  $\mathfrak{P}^B$ by removing all B-hyperedges $E=(E^-,\{u\})$ for which
  $u\in S$ or there is a hyperedge $E=(F^-,\{u\})$ with $F\preceq_B E$.
  Clearly, $E(\mathscr{P}^*)$ is partially ordered and each of its vertices
  is reachable from $S$. Thus $\mathscr{P}^*$ is grounded and acyclic, and
  therefore an assembly pathway on the B-hypergraph
  $\mathscr{H}^B$. Conversely, for an assembly pathway on the B-hypergraph
  $\mathscr{H}^B$ we induce a partial order on the general hypergraph
  $\mathscr{H}$ by setting $E\prec F$ for two distinct edges $E$ and $F$ of
  $\mathscr{H}$ if there are vertices $z\in E^+$ and $y\in F^+$ such that
  $(E^-,\{z\})\prec_B (F^-,\{y\})$, i.e., if $y$ in $F$ is reachable only
  after $z$ in $E$. Since the edge set $E(\mathfrak{P}^B)$ is partially
  ordered, we obtain a partial order on a subset $E(\mathscr{P})$ of
  hyperedges on $\mathscr{H}$ such that for each $z\in F^-$ there is a
  hyperedge $E\prec F$ with $z\in E+$, and thus $\mathscr{P}$ is an
  assembly pathways according to Definition~\ref{def:genass}.
\end{proof}
The point of Lemma~\ref{lem:general} is to show that assembly can
meaningfully be expressed in general directed hypergraphs with reactions
that produce more than one product, and thus in particular pertains to
chemical reactions networks in full generality.

\subsection{Alternative Measures of Assembly Complexity} 

We now turn to the problem of defining an assembly index in the setting of
general directed hypergraphs. We argue that there are at least two equally
meaningful ways that do not yield identical results. In the setting of
chemical synthesis, each reaction product that is further use, i.e.,
appears in an $E^-$, must be isolated separately and hence incurs its own
cost, even if it is produced together with another useful product. This
scenarios argues for counting the vertices appearing in $E^-$, or,
equivalently the B-hyperedges in $\mathscr{P}^B$. In chemical evolution, on
the other hand, the fortuitous production of multiple relevant assembly
intermediates incurs no extra cost, and hence natural measure is the
cardinality of $E(\mathscr{P})$.
\begin{proposition}
  Let $\mathscr{P}$ and $\mathscr{P}^B$ be minimal assembly pathways for
  $x$ on $\mathscr{H}$ and $\mathscr{H}^B$, respectively. Then
  $\min|E(\mathscr{P})|\le \min |E(\mathscr{P}^B)|$.
\end{proposition}
\begin{proof}
  Let $\mathscr{P}$ be a minimal pathway from $S$ to $x$ in $\mathscr{H}$,
  and let $\mathscr{P}^B$ be the corresponding pathway in $\mathscr{H}^B$,
  i.e., $(E^-,\{x\})\in E(\mathscr{P}^B)$ if $x\in E^+$ and $E=(E^-,E^+)\in
  E(\mathscr{P})$. If $\mathscr{P}$ is a minimal pathway from $S$ to $x$,
  then each of its hyperedges $E$ has a vertex $y_E\in E^+$ that is not
  contained in $S$ or a hyperedge $E'\prec E^+$, since otherwise $E$ can be
  removed, contradicting minimality. Therefore, the hyperedge
  $E_y=(E^-,\{y\})$ must be contained in a minimal pathway obtained by
  hyperedge-removal from $\mathscr{P}^B$. Conversely, every minimal
  assembly pathway from $S$ to $x$ in $\mathscr{H}^B$ can be extended to a
  subhypergraph of $\mathcal{H}$ by inserting, for every hyperedge
  $(E,\{x\})$ a corresponding hyperedge in $\mathscr{H}$ from which it
  derived. Thus minimal assembly pathways in $\mathscr{H}^B$ cannot be
  shorter than their corresponding assembly pathway in $\mathscr{H}$.
\end{proof}
In $\mathscr{H}$, however, we may exploit the fact that more than one
vertex in $E^+$ may be utilized later, thus reducing the number of required
edges, i.e., assembly steps. As a consequence, the minimal number of
vertices and the minimal number of hyperedges no longer matches for
assembly in general directed hypergraphs. We will briefly comment on issues
of computational complexity in Section~\ref{ssect:graphass}.

Counting the minimal number of distinct hyperedges, i.e., reactions,
however is not the only meaningful way of associating a cost to a pathway
$\mathscr{P}$ in $\mathscr{H}$.  For applications in synthesis planning, a
different measure is more natural \cite{Hendrickson:77}. In the setting of
B-hypergraphs, every $u\in E_k^-$ is the head of a unique edge $E(u)$ or
$u\in S$. In the first case, the cost of producing $u$ is the total cost
$c(u)$ of the sub-B-hyperpath from $o$ to $u$, i.e., the cost of the
synthesis of $u$. In the second case there is the cost $c(s)$ of starting
material. The cost of the synthesis of $v$ is composed of cost of the
materials $u\in E^-(v)$ plus some cost $c^*(v)$ for the synthesis step
$E(v)$ itself, see e.g.\ \cite{Nielsen:05,Fagerberg:18a}. This yields the
recursions
\begin{equation*}
  c(v) = c^*(v)+\sum_{u\in E(v)^-} c(u) \,.
\end{equation*}

In the simplest case, only the number of reactions, i.e., hyperedges is
counted.  This amounts to setting $c(s)=c^*(s)=0$ for $s\in S$ and
$c^*(v)=1$ for all $v\in V\setminus S$. Then $c(v)$ simply counts each
non-augmentation edges along the sub-B-path $\mathscr{Q}_v$ from $o$ to $v$
as many times as its head appears in the tail of another hyperedge in
$\mathscr{Q}_v$. Therefore, $c(v)$ equals the number of hyperedges in
$\mathscr{Q}_v$ if and only if every head is used only once. Denoting by
$\hat c(x)$ the length of the shortest B-hyperpath from $o$ to $x$ (not
counting the augmentation edges), we obtain immediately that $\ass(x)\le
\hat c(x)$ with equality if and only if no head is used more than once.
Equality thus holds in particular in the ``split-branched'' spaces
introduced in \cite{Marshall:22}.  The shortest B-hyperpath, and thus in
particular also its ``length'' $\hat c(x)$ can be computed efficiently by
means of dynamic programming \cite{Nielsen:05}. It is worth noting in this
context that B-hyperpaths with several alternative cost functions have been
considered as ``minimum makespan assembly plans'' already in
\cite{Gallo:02}.

\section{Computational Considerations} 
\label{sect:comp}

\subsection{Assembly by Graph Rewriting}
\label{ssect:graphass}

Up to this point we have not considered the rule-based aspect of the
assembly theory and instead only focused on the consequences of rule
application of a seed set $S$, namely the resulting directed hypergraph
that captures an abstract temporal dimension of the formalism.  In the context
of chemistry, it is natural to modeling assembly in terms of graph
transformations since molecules have a natural representation as (labeled)
graphs and chemical reaction mechanisms determine specific patterns of
rearrangements of chemical bonds, i.e., edges \cite{Andersen:17b}.

Chemical graph transformation are naturally modeled by means of
double-pushout graph rewriting \cite{Andersen:13a}. A reaction rule is
specified as a span, i.e., a pair of graph monomorphisms $L\leftarrow
K\rightarrow R$, where $K$ is a graph describing parts that remain unchanged,
while $L$ and $R$ are the corresponding patterns in reactants and products,
respectively. In application to chemistry, the correspondence respects the
preservation of atoms, i.e., the restriction of two graph monomorphisms to
the vertices are bijective. The application of such a rule requires 
that $L$ appears a subgraph in the substrate graph. The match of $L$ in $G$
is then replaced by $R$ in a principled manner, resulting in a graph $H$ of
reaction products. The connected components of the reactant and product
graphs $G$ and $H$ represent the molecules. Each rule application, formally
called a derivation $G\Rightarrow H$ thus defines hyperedge with the
molecules of $G$ and $H$ as its tail and head, respectively. The underlying
hypergraph $\mathscr{H}$, therefore can constructed explicitly by
iterative rule application. For technical details we refer to
\cite{Andersen:13a}.

A key observation is that each DPO rewriting rule $r=(L\leftarrow
K\rightarrow R)$ has a natural inverse $r^{-1}=(R\leftarrow K\rightarrow
L)$. Instead of constructing $x$ from $S$ using a given rule set, one can
equivalent decompose $x$ iteratively into constituents until all of them
are reduced to building block in the seed set $S$.  Assembly in the sense
of Cronin and Walker is therefore in essence the inverse of
retrosynthetic analysis \cite{Vleduts:63,Corey:67}, where chemical
structures are recursively split into smaller components until sufficiently
simple starting materials are reached. Similar ideas have been explored in
synthesis planning \cite{Bertz:93}.

The invertibility of graph rewriting rules implies that the assembly as
defined by \cite{Marshall:21}, and on rule-derived hypergraphs in general,
can be computed by first expanding a space of all possible compounds
obtained by \emph{disassembling} the focal compound $x$, and then searching
for an optimal assembly pathway in the resulting hypergraph. The expansion
of the space can be defined as the iterative application of a single simple
graph rewriting rule. This rule selects a single vertex in the molecular
graph and duplicates it, creating two vertices without an edge in
between. Semantically, this operation either splits a compound into two or
breaks a cycle. The formal DPO rule is provided in the appendix, and can be
seen more abstractly in Table~\ref{tab:daniel} below.  Notably, this
transformation always increases the number of atoms in the products
relative to the reactant, which from a chemical (retro-)synthetic
perspective might be considered controversial. The process is applied
recursively until only base compounds, consisting of two vertices and one
edge, remain. The number of necessary splitting reactions defines the
assembly index. The ring-breaking application does not increase the
assembly index. The resulting acyclic hypergraph, when inverted, serves as
a witness structure demonstrating an optimal assembly index and, as
proposed by Cronin and Walker \cite{Marshall:21}, can be interpreted as
an assembly pathway reflecting molecular complexity.

Optimal assembly pathways can be identified using different approaches. One
can use graph transformation frameworks such as \texttt{M{\O}D}
\cite{Andersen:16a} and integer hyperflow methods, alternatively through
integer linear programming (ILP), as detailed in the appendix. Recent work
\cite{Kempes:24,Lukaszyk:25} has established that finding an optimal
assembly pathway, more precisely the Assembly Steps Problem (see also
\cite{Lukaszyk:25}) is NP-complete, while the complexity of Assembly Index
Problem remains open.

\subsection{Cyclization Effects}

In any witness for a specific assembly index, all cyclizations can be
performed as the last steps since, by definition, they do not contribute to
the assembly index. The definition in \cite{Marshall:21} also treats ring
closures by means of vertex identification, which differs from the
conventional chemical notion of cyclization. The assembly pathway
hypergraph in Table~\ref{tab:daniel} demonstrates this effect: the last
directed edge accounts for all five cyclization reactions needed to form
cubane.

\begin{figure}[h]
  \centering
  \includegraphics[width=\textwidth]{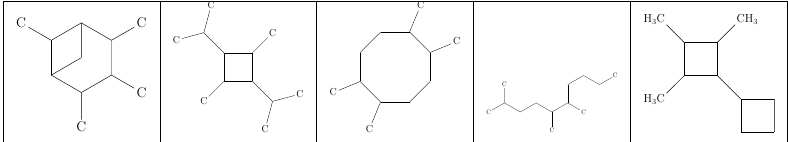} 

  \caption{Molecular structures resulting from iterative inverted
    cyclization to cubane; 4 leftmost: all of them have (like cubane) a
    molecular assembly index of 4; rightmost: molecular assembly index of
    5}
  \label{fig:cubane-table}
\end{figure}

Since cyclizations in a witness pathways do not contribute to the assembly
index, many chemically distinct compounds share the same index. For
example, cubane has a cyclomatic number of five, requiring five vertex
identifications to close rings. Applying only de-cyclization reactions to
cubane yields 241 compounds, 141 of then are part of a witness hypergraph
for the optimal assembly index, all with an assembly index of
4. Figure~\ref{fig:cubane-table} depicts four representative examples,
including a completely acyclic compound. The resulting compounds are
extremely heterogeneous from a chemist's point of view, an observation that
is not specific to the example of cubane but pertains to all cyclic
molecules. Out of the 241 possible de-cyclization products of cubane, 100
result in molecular structures with an assembly index of $5$, showing that
single cyclizations events can decrease the assembly index. The rightmost
molecular structure in Figure~\ref{fig:cubane-table} serves as an
example. Notably, de-cyclization of this specific structure can yield
molecular structures with an assembly index of 4, showing that single
cyclizations events can increase the assembly index. This highlights a
problematic aspect of the assembly index measure: a single cyclization can
either decrease, increase, or leave the assembly index unchanged. For a
complexity measure based on counting discrete events, this behavior is, at
best, counter-intuitive. In contrast, more traditional approaches assign a
cost to cyclization, as demonstrated by the dynamic programming approach.

From a chemical perspective, cycles significantly contribute to the
synthetic complexity of compounds. As integral parts of the overall
architecture of molecules they heavily impact the stereochemistry and
the practical difficulty of synthesis.  Cyclic systems, in particular small
rings, introduce significant ring strain, due to deviations from ideal bond
angles, and thus influence the reactivity. The spacial arrangement of
substituents is also effected by cycles due to their specific
conformational preferences, making stereochemical control hard during
synthesis. Installing functional groups onto cyclic structures requires
sophisticated synthetic strategies, due to steric hindrance and the
conformational constraints of rings. These chemical arguments imply that
cycles cannot be be ignored in chemical complexity measures. In a series of
theoretical studies on the optimality of synthesis plans, Hendrickson
\cite{Hendrickson:77} already showed in the 1970s that the placement of
cyclizations in a synthesis plan has a strong influence on its optimality;
good synthesis plans perform cyclizations early. The worst placement in
general is towards the end of the synthesis plan. Therefore, ignoring cyclizations in
measures of assembly complexity does not seem advisable. Excluding such reactions in the definition of the the assembly
index in \cite{Marshall:17,Marshall:21,Marshall:22,Liu:21,Jirasek:23} thus
seems hard to justify from both the chemical and the mathematical point of
view.

\subsection{Edge Removal and Alternative Assembly Measures}

Edge-based disassembly has a long-standing tradition in synthesis planning
\cite{Bertz:93}. Here, instead of merging vertices, compounds are
iteratively decomposed by removing edges. Both approaches can be compared
directly. Table~\ref{tab:daniel} shows results for cubane and the
pyrrolidine dimer \textit{1-(2-Pyrrolidinylmethyl)pyrrolidine} (ChemSpider
ID: 126001) using the vertex splitting/identification rule, while
Table~\ref{tab:daniel2} presents results using the edge removal/adding
rule.

In addition to assembly index, we explore alternative quality measures for
synthesis planning. Hendrickson \emph{et al.}\ \cite{Hendrickson:77}
proposed that optimality should be measured in terms of convergence:
synthesis plans that minimize reaction steps and balance precursor usage
tend to minimize material loss. The total weight of starting materials (TW)
serves as a measure of synthetic efficiency. This is computed recursively,
assigning weights based on precursor costs and a retro-yield parameter
\cite{Fagerberg:18a}. Unlike the assembly index, the total weight of
starting materials (TW) allows for using dynamic programming (DP). The TW
measure defines the cost of assembling a compound as the sum of the costs
of its educts, adjusted by a retro-yield factor. Specifically, if a
compound $m$ is formed by affixing compounds $m_1$ and $m_2$, its cost is
given by the sum of the costs of $m_1$ and $m_2$, multiplied by a
retro-yield parameter $r>1$, see \cite{Fagerberg:18a} for details. While
the assembly cost of a compound cannot be directly computed as the sum of
the assembly costs of its educts due to potential overlaps in assembly
spaces, the TW measure remains well-defined and can be efficiently computed
using dynamic programming \cite{Fagerberg:18a}.

Tables~\ref{tab:daniel} and~\ref{tab:daniel2} compare optimal pathways
according to the assembly index (light red background) and the TW objective
(light green background).  For cubane, there are 412 optimal DP solutions,
all requiring at least four affixations. The optimal assembly index pathway
in the depicted witness requires 4 affixations and 5 additional
cyclizations, leading to a less convergent synthesis plan. The longest path
from root to sink under the assembly index is 9 steps, whereas all 412
DP-optimal solutions require at most 8 steps.  For the pyrrolidine dimer,
645 DP-optimal synthesis plans exist, each requiring at least 6
affixations. The optimal assembly index pathway, however, requires only 5
affixations but produces a significantly less convergent synthesis.

Using the edge removal rule, Table~\ref{tab:daniel2}, a similar pattern
emerges. For cubane, the longest path using assembly index rules is 3
affixations plus 5 cyclizations in the depicted witness, whereas all
optimal DP solutions result in an earlier cyclization step, yielding 3
affixations and 4 cyclizations. For the pyrrolidine dimer, DP optimization
requires one additional affixation but reduces the longest path from source
to sink from 7 to 5 steps. 

\begin{table}[p]
\enlargethispage{20pt}
\centering
\includegraphics[width=\textwidth]{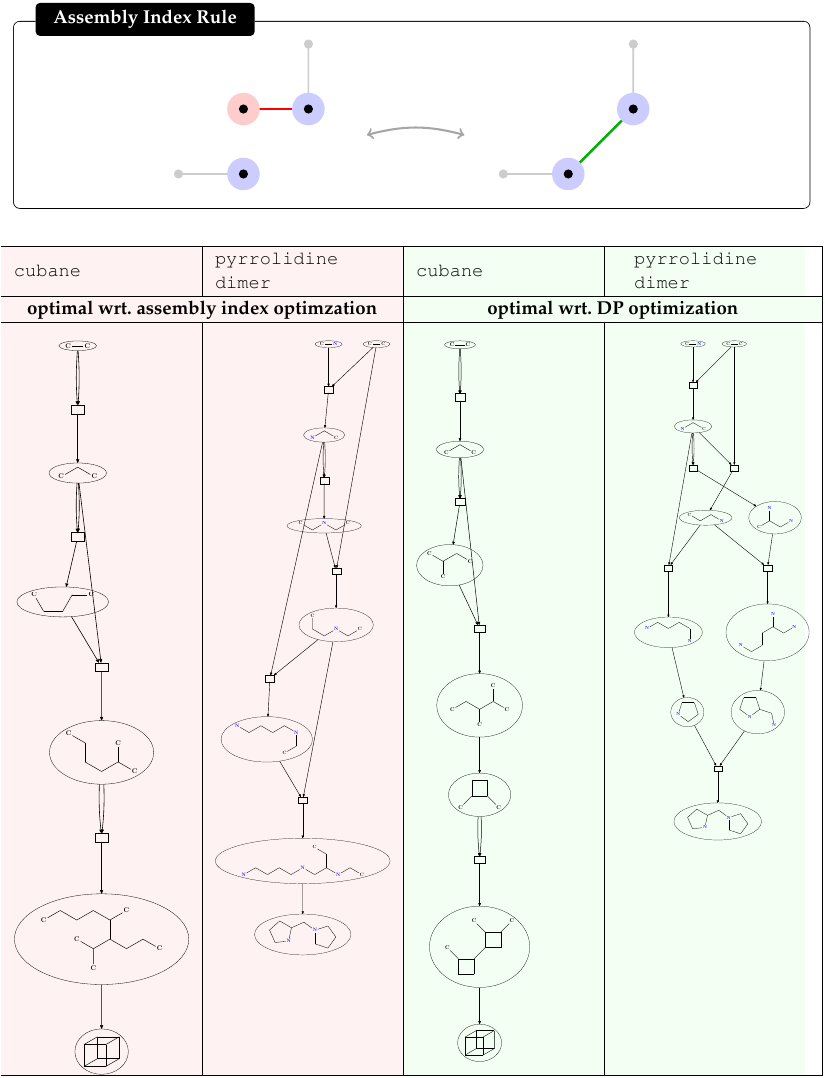}

\caption{\small Synthesis plans for cubane and
  1-(2-Pyrrolidinylmethyl)pyrrolidine (ChemSpider ID: 126001), based on the
  assembly index rule shown at the top. In this rule, a red node and a blue
  node share the same label and are identified for merging; the red node
  and its associated red edge are then removed, and the two blue nodes are
  connected by a new edge. The left two columns (light red) show plans that
  are optimal with respect to assembly index optimization but suboptimal
  with respect to DP optimization, while the right two columns (light
  green) depict plans optimal with respect to DP optimization but
  suboptimal with respect to assembly index optimization. Note, that chains
  of cyclization reactions are represented as a single edge—for example, in
  the leftmost plan, 12 cyclization reactions are bundled into one
  reaction. The assembly indices of the four depicted plans correspond to
  the number of hyperedges reflecting merging steps: 4, 5, 5, and 6,
  respectively.}
\label{tab:daniel}
\end{table}

\begin{table}[p]
\enlargethispage{20pt}
\centering
\includegraphics[width=\textwidth]{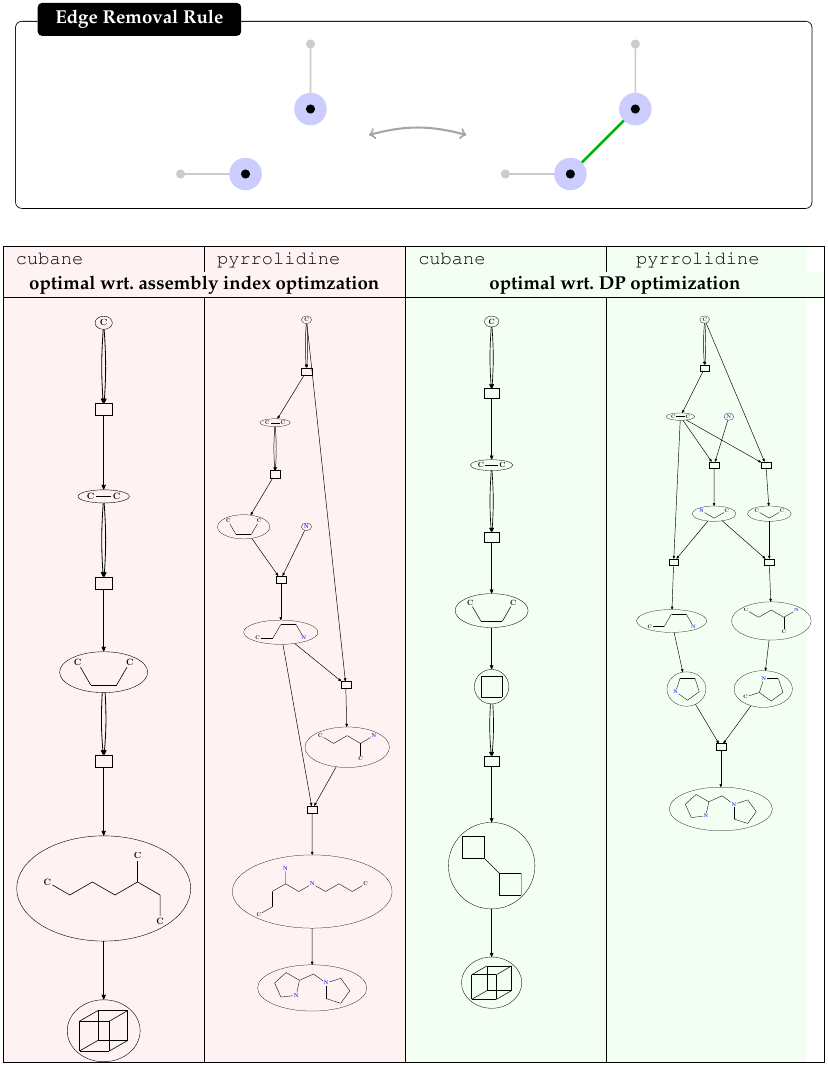}

\caption{\small Synthesis plans for cubane and
  1-(2-Pyrrolidinylmethyl)pyrrolidine (ChemSpider ID: 126001), based on
  edge addition/removal rule. In this rule (shown on top), the two blue
  nodes are simply connected by a new edge. In the left two columns only
  affixations contribute to the objective. Similar to the results in
  Table~\ref{tab:daniel}, not accounting cylizations in the objective
  measure (light red) leads to chemically questionable witnesses, whereas
  accounting for cylization costs --- as done when using the DP approach
  --- leads to chemically more relevant witnesses. Note, that e.g. last
  step depicted with just a single arrow in the left-most synthesis for
  cubane reflects 5 cyclization reactions, which do not contribute to the
  objective value (3 as there are 3 affixation reactions). As can be seen
  in the green columns: an optimal DP solution enforces and earlier
  cyclization. }
\label{tab:daniel2}
\end{table}

\subsection{Reachability} 

The vertex construction and edge insertion rules considered above obviously
can generate all graphs from minimum size building blocks. In the case of
labeled graphs the necessary diversity of labeled building block has to be
supplied in the seed sets $S$. The question whether an assembly pathway
from $S$ to $x$ exists in general boils down to a (usually difficult)
hyperpath problem. Starting from the rules, however, it is already a
difficult problem to determine whether $x$ is reachable from $S$ by a
finite set of rule applications, or converse, whether $x$ can disassembled
completely into elements of $S$. In fact, it is undecidable whether as a DPO
graph rewriting systems terminates \cite{Plump:98}.

To highlight the influence of the rule, we consider a gluing operation that
requires the presence of one or more finite patterns graphs $Q$. 
If two graph $G_1$ and $G_2$ contain $Q$ as an
isomorphic subgraph, we may \emph{glue} $G_1$ and $G_2$ together by
identifying vertices along subgraph isomophisms or graph monomorphisms
$Q\to G_1$ and $Q\to G_2$.  This operation is used e.g.\ in the context of
graph alignments for common induced subgraphs $Q$
\cite{GonzalesLaffitte:23a}. In the setting of graph monomorphisms this
gluing operation can be expressed as the DPO graph rewriting rule
\begin{equation*}
  \gamma_Q \coloneqq (2Q\leftarrow 2Q \rightarrow Q)
\end{equation*}
where $2Q\coloneqq Q\cupdot Q$ is the disjoint union of two copies of $Q$
and the graph monomorphism $\rho$ is given by by $\rho(Q_1)=Q$ and
$\rho(Q_2)=Q$ for two copies of $Q$ in $2Q$.  Note that in the context of
graph transformation, a rule can be applied both to a connected and a
disconnected substrate graph. It is not at all obvious in this case which
graphs are reachable by such a gluing operation. 

In special cases it is possible to guarantee that reachability can be
checked easily. This is in particular the case if there is a graph
invariant $\eta$ with a total order $\lessdot$ such that, for each edge
$(E^-,E^+)$ in $\mathscr{H}$ we have $\eta(y)\lessdot\eta(z)$ for all $y\in
E^-$ and $z\in E^+$. In terms of the rules, we therefore require that the
application of a rule to any substrates increases $\eta$. The total order
may e.g.\ be a lexicographic, as in the case were a reaction either
increases the number of vertices of the cyclomatic number as in the case of
edge insertion. In this scenario it suffices to explore the monotonically
$\lessdot$-decreasing paths originating from $x$ and terminating in some
$z\in S$.

\section{Concluding Remarks}

Assembly theory has been discussed intensively in response to its claim to
providing a unified approach to understanding complexity and evolution in
nature \cite{Sharma:23}. Here, we do not delve into the controversies about
the empirical validity or accuracy of assembly as a measure of complexity,
the degree of causation, or its capability of distinguishing biogenic from
abiogenic molecules \cite{Abel:25,Jaeger:24,Hazen:24}. Instead, we have
analyzed the formal framework of assembly theory as proposed by Cronin,
Walker and collaborators \cite{Marshall:17,Marshall:21,Liu:21,Sharma:23}
and showed that the approach can be understood as a specific case of a
hyperpath problem in directed hypergraphs. This not only links the assembly
index to a broader body of literature on related computational problems but
also leads to generalizations that apply to chemical reaction networks in
particular and graph-transformation systems in general. The assembly index
as defined in \cite{Marshall:22} can indeed be expressed by means of a very
simply graph rewriting rule. The general formal framework introduced here
suggests a multitude of alternative measures that are as natural as the
original definition.

The formal framework introduced here also highlights the close connections
between assembly theory with retrosynthesis planning and computational
approaches in synthesis planning that have been present in the chemical
literature since the seminal work of Vl{\'e}duts \cite{Vleduts:63} and
Corey \cite{Corey:67} in the 1960s. This suggest to consider measures
commonly used in synthesis planning that are closely related to the assembly
index but are considerably easier to compute.

These results suggest a conceptual analogy to the problem of finding a
minimum cycle basis (MCB) \cite{Horton:87,dePina:95}. Historically, cycle
basis identification in chemical graph analysis often relied on spanning
tree heuristics, despite the fact that optimal MCB identification based on
a spanning tree (i.e.\ finding an optimal Kirchhoff fundamental cycle basis)
is NP-complete \cite{dePina:95}. A more efficient approach, polynomial-time
solvable, allows MCBs to be computed without relying on spanning
trees. Similarly, in synthesis planning, the choice of expansion rules and
objective functions significantly influences solution quality. The DP-based
approach described here is computationally more efficient. It is more
closely related to the ``split-branched assembly index'' used in
\cite{Marshall:21,Marshall:22}. More importantly, it yields more chemically
meaningful solutions than assembly index-based methods.  Future research
thus could explore hybrid strategies that leverage both techniques to
refine synthesis pathway predictions further.

\paragraph{Acknowledgements.} This work was funded by the Novo Nordisk Foundation as part of MATOMIC
  (0066551). PFS acknowledges the financial support by the Federal Ministry
of Education and Research of Germany (BMBF) through DAAD project 57616814
(SECAI, School of Embedded Composite AI), and jointly with the
S{\"a}chsische Staatsministerium f{\"u}r Wissenschaft, Kultur und Tourismus
in the programme Center of Excellence for AI-research \emph{Center for
Scalable Data Analytics and Artificial Intelligence Dresden/Leipzig},
project identification number: SCADS24B.


\vskip2pc

\bibliographystyle{abbrv}
\bibliography{assx}

\newpage

\section{Appendix}

This appendix introduces a graph transformation rule that expands the set
of all compounds reachable through iterative rule application until
convergence. Based on the resulting directed hypergraph, we present a
simple ILP formulation to identify witnesses with an optimal molecular
assembly index by minimizing inverse affixations. An iterative approach for
enumerating all such optimal witnesses is straightforward: rerun the ILP
while adding constraints that forbid any solutions already found. Using
this method, we can characterize all witnesses that achieve an optimal
molecular assembly index and, for example, demonstrate that none of these
witnesses also attains optimality with respect to more chemically motivated
objective functions, as discussed in the main body of the paper.

\subsection{DPO diagrams}
In Figure \ref{fig:dpo} we depict the Double Pushout (DPO) diagram for the
rule that was used to model the inverse affixation and the inverse
cyclization as used by Cronin, Walker and colleagues. Note that the
rule application might either split a compound into two compounds (inverse
affixation) or break a ring (inverse cyclization). In the top span $\_X$ and
$\_Y$ are variables that can match any atom type. To reduce clutter we omit
the variable of the bond type, but any bond type can match the bond between
atom $\_X$ and $\_Y$.

\begin{figure}[h]
  \centering
  \includegraphics[width=0.9\linewidth]{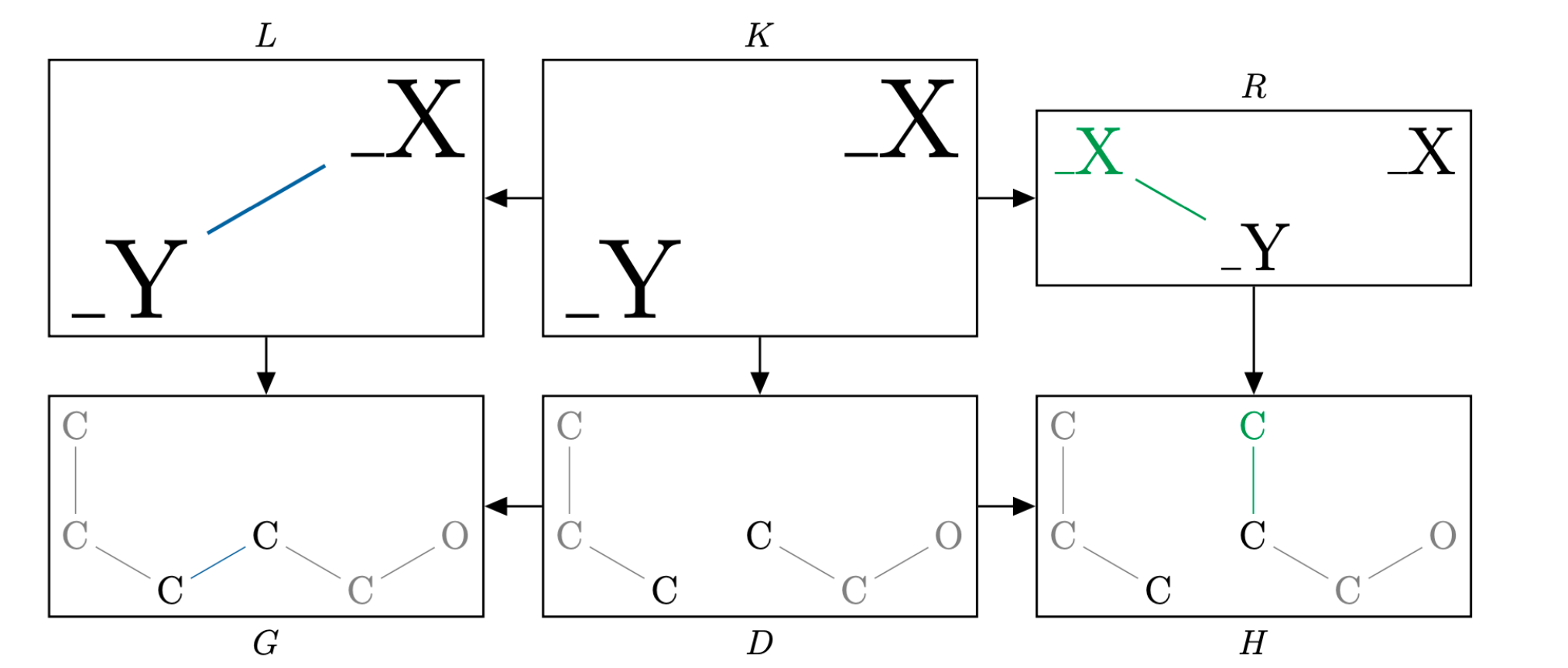}
  \caption{DPO diagram, DPO Rule (top span) and an exemplification of its
    application to butanol - butanol is split into two compounds, leading
    to an inverse affixation (however, note, that the number of carbons
    does not remain invariant); the rule is the only rule needed to expand
    the chemical space which is expanded to find optimal solutions w.r.t.
    optimal assembly index}
        \label{fig:dpo}
    \end{figure}

\subsection{ILP Formulation}
In the following we describe an ILP approach to solve for the optimal
molecular assembly index as described by Cronin, Walker and colleagues.

\noindent
Let $\mathscr{H} = (V,E)$ be a directed hypergraph built via repeated
application of the DPO rule from Figure~\ref{fig:dpo}, ensuring that all
target compounds have at most two vertices. Each vertex $v \in V$
represents a compound in this construction; each hyperedge $e \in E$ stands
for an inverse reaction step (either an inverse cyclization or an inverse
affixation). We assign to every hyperedge $e$ a binary weight $w_e \in
\{\,0,1\}$, where $w_e = 0$ indicates an inverse cyclization and $w_e = 1$
indicates an inverse affixation.

We introduce a binary variable $y_v \in \{0,1\}$ for each vertex $v\in V$,
where $y_v = 1$ means $v$ is chosen in the ``optimal witness.'' Likewise,
for each hyperedge $e\in E$, we introduce a binary variable $x_e \in
\{0,1\}$, indicating whether $e$ is used in the solution. We distinguish a
special goal vertex $g \in V$ corresponding to the target compound, and fix
$y_g = 1$.

Our objective is to minimize the total count of inverse affixations (where
$w_e=1$) and also prefer fewer total hyperedges. Thus we set
\begin{equation*}
\min\;\; \sum_{e\in E} w_e.
\end{equation*}
In order to select for simpler optimal witnesses, we may use the following
cost function:
\begin{equation*}
\min\;\; \sum_{e\in E}\bigl(1000\,w_e \;-\;1\bigr)\,x_e.
\end{equation*}
Here $w_e=1$ for inverse affixations and $w_e = 0$ for inverse
cyclizations. The $-1$ favors smaller sets of hyperedges among equally good
solutions in terms of affixations.

We enforce a collection of constraints that ensure consistency. First, for
each vertex $v$ we denote by $E_v^{\mathrm{out}}$ the hyperedges for which
$v$ is a source. We require
\begin{equation*}
\sum_{e\in E_v^{\mathrm{out}}} x_e \;=\; y_v, 
\quad
\sum_{e\in E_v^{\mathrm{out}}} x_e \;\le\; 1,
\quad
y_v \;\ge\; x_e\;\;\forall\,e\in E_v^{\mathrm{out}},
\end{equation*}
so that each chosen vertex has exactly one outgoing hyperedge in the
witness, and we never pick more than one hyperedge leaving the same
vertex. If $x_e = 1$, the source vertex must also be chosen ($y_v \ge
x_e$). Next, for each vertex $v$ with set $E_v^{\mathrm{in}}$ of incoming
hyperedges, we impose
\begin{equation*}
\sum_{e\in E_v^{\mathrm{in}}} x_e \;\ge\; y_v,
\end{equation*}
ensuring that if a vertex is in the witness, it is supplied by at least one
incoming hyperedge. We also have pairs $(e,v)$ where $v$ must be active if
$e$ is chosen, enforced by $y_v \ge x_e$. Finally, the condition $y_g = 1$
fixes the goal vertex as part of the solution. Altogether, this integer
linear program encodes a minimum-affixation hyperpath from the initial
sources to the goal compound, with cyclizations and affixations as
described by $E$.

\end{document}